\documentclass[letter]{ieice}

\usepackage{graphicx,xcolor}
\usepackage[fleqn]{amsmath}
\usepackage{newtxtext}
\usepackage[varg]{newtxmath}

\usepackage{enumitem}
\usepackage{tikz}
\usepackage{url}

\field{D}
\title{A Subclass of Mu-Calculus with the Freeze Quantifier
  Equivalent to B\"uchi Register Automata}
\titlenote{%
  This work was partially supported by JSPS KAKENHI Grant Numbers
  22H03568, 23K24824 and 24K14901.}
\authorlist{%
  \authorentry[takata.yoshiaki@kochi-tech.ac.jp]{Yoshiaki TAKATA}{m}{kochi}%
  \MembershipNumber{9618813}
  \authorentry{Akira ONISHI}{n}{nagoya}%
  \authorentry{Ryoma SENDA}{n}{mie}%
  \authorentry[seki@is.nagoya-u.ac.jp]{Hiroyuki Seki}{f}{nagoya}%
  \MembershipNumber{8206514}
}
\affiliate[kut]{School of Informatics,
 Kochi University of Technology,
 Tosayamada, Kami City, Kochi 782-8502, Japan}
\affiliate[nagoya]{Graduate School of Informatics,
 Nagoya University,
 Furo-cho, Chikusa, Nagoya 464-8601, Japan}
\affiliate[mie]{
 Graduate School of Engineering, Mie University,
 Kurima\-machiya-cho, Tsu City, Mie 514-8507, Japan}

\received{2015}{1}{1}
\revised{2015}{1}{1}

\newtheorem{definition}{Definition}
\newtheorem{theorem}{Theorem}
\newtheorem{lemma}{Lemma}

\newtheorem{example}{Example}
\newenvironment{proof}{\par\noindent\emph{Proof.}\ }{}

\usetikzlibrary{automata,shapes}
\tikzset{
  ->, 
  >=latex, 
  auto, 
  initial text={}, 
}

\newcommand{\LTLfo}{{\mathrm{LTL}^\downarrow}_{\mathrm{d}\omega}}
\newcommand{\muf}{\mu^\downarrow}
\newcommand{\mufd}{{\mu^\downarrow}_{\!\mathrm{d}}}
\newcommand{\mufo}{{\mu^\downarrow}_{\!\mathrm{d}\omega}}

\newcommand{\TT}{\mathtt{tt}}
\newcommand{\FF}{\mathtt{ff}}

\newcommand{\TO}{\mathbin{\rightarrow}}
\newcommand{\Gets}{\mathbin{\leftarrow}}
\newcommand{\calA}{\mathcal{A}}

\newcommand{\calM}{\mathcal{M}}

\newcommand{\Var}{\mathit{Var}}
\newcommand{\Varo}{\mathit{Var}_{\omega}}
\newcommand{\omegavariable}{$\omega$-variable}
\newcommand{\Vo}{V_{\omega}}

\newcommand{\Vtt}{V_{\TT}}
\newcommand{\At}{\mathit{At}}

\newcommand{\Env}{\mathit{Env}}

\newcommand{\lfp}{\mathrm{lfp}}
\newcommand{\Str}[1]{\mathrm{fst}(#1)}
\newcommand{\Snd}[1]{\mathrm{snd}(#1)}

\newcommand{\Fsigmaw}{F_{\sigma\!,w}}
\newcommand{\Fsigmaiw}{F_{\sigma_1,w}}

\newcommand{\FsigmaIwi}{F_{\sigma_i,w'}}

\newcommand{\Nat}{\mathbb{N}}

\newcommand{\powerset}[1]{\mathcal{P}(#1)}
\newcommand{\Dat}{\mathbb{D}}

\newcommand{\neXt}{\mathord{\mathsf{X}}}
\newcommand{\Until}{\mathbin{\mathsf{U}}}
\newcommand{\WeakUntil}{\mathbin{\mathsf{W}}}
\newcommand{\Global}{\mathord{\mathsf{G}}}

\newcommand{\Down}{\mathord{\downarrow}}
\newcommand{\Up}{\mathord{\uparrow}}

\newcommand{\Rangew}[2]{w^{[{#1}\mathbin{:}{#2}]}}

\newcommand{\Buchi}{B\"uchi}

\begin{document}
\maketitle
\begin{summary}
Register automaton (RA) is an extension of
finite automaton for dealing with data values
in an infinite domain.
RA has good properties such as
the decidability of the membership and emptiness problems.
In the previous work,
we proposed
disjunctive $\muf$-calculus ($\mufd$-calculus),
which is a subclass of
modal $\mu$-calculus with the freeze quantifier,
and showed that it has the same expressive power as RA\@.
However,
$\mufd$-calculus is defined as a logic
on finite words,
whereas temporal specifications in model checking are
usually given in terms of infinite words.
In this paper,
we re-define the syntax and semantics of $\mufd$-calculus
to be suitable for infinite words and
prove that the obtained temporal logic,
called $\mufo$-calculus,
has the same expressive power as \Buchi\ RA\@.
We have proved the correctness of
the equivalent transformation
between 
$\mufo$-calculus
and \Buchi\ RA
with a proof assistant software Coq.
\end{summary}

\section{Introduction}
\label{sec:intro}

Register automaton (RA)~\cite{KF94} is an extension of
finite automaton (FA)
that has registers
for dealing with data values.
RA preserves good properties of FA such as
the closure property and
the decidability of the membership and
emptiness problems~\cite{SI00,NSV04}.
RA has been used as a formal model
in 
verification of systems~\cite{GDPT13} and
reactive synthesis~\cite{EFR21}.
Model checking is a technique
to verify whether every run of a model $\calM$
of a system
satisfies a given specification $\varphi$.
For a finite-state model of a system and a linear temporal logic (LTL)
specification,
the model checking
can be reduced to the emptiness problem of
a \Buchi\ automaton~\cite{CGKPV18}.
LTL model checking has also been shown to be decidable
for pushdown systems (PDS)~\cite{EKS03}
and register pushdown systems (RPDS)~\cite{STS21-IEICE-dec}.
These model-checking methods rely on the fact that
every LTL formula can be transformed into
an equivalent \Buchi\ automaton.
To extend these model-checking methods to models with registers,
we need an appropriate temporal logic to give a specification,
which can be transformed into an equivalent \Buchi\ RA\@.

In the previous work,
we proposed
disjunctive $\muf$-calculus
($\mufd$-calculus)~\cite{TOSS23},
which is a subclass of
modal $\mu$-calculus with the freeze quantifier
($\muf$-calculus)~\cite{JL07}
and is shown to have the same expressive power as RA\@.
Since the class of RA is not closed under complement
and an RA does not have the ability to simulate logical conjunction,
$\mufd$-calculus 
imposes restrictions on the usage of negation and conjunction.
On the other hand,
by a system of equations
for representing recursive properties,
$\mufd$-calculus can
represent an arbitrary cycle in an RA,
which cannot be represented by
temporal operators such as
$\neXt$ (next),
$\Until$ (until) and $\Global$ (global) in LTL\@.
We can say that $\mufd$-calculus is
one of the most expressive temporal logic
that can be equivalently transformed into RA\@.
However, $\mufd$-calculus is defined as a logic
on finite words, 
whereas temporal specifications in model checking are
usually given in terms of infinite words.

In this paper,
we re-define the syntax and semantics of $\mufd$-calculus
to be suitable for infinite words, and we
prove that the obtained temporal logic,
called $\mufo$-calculus,
has the same expressive power as \Buchi\ RA\@.
We have proved the correctness of
the equivalent transformation
between 
$\mufo$-calculus
and \Buchi\ RA
with a proof assistant software Coq.

The main difference between
$\mufd$-calculus and $\mufo$-calculus is that
the former requires only
the least solution of a system of equations
while the latter has to consider both
the least and greatest solutions of systems of equations.
For example,
a proposition
``every (resp.\ some)
odd position in a given infinite word
satisfies a given atomic proposition $p$,''
which can be represented by \Buchi\ RA,
is represented
in $\mufo$-calculus as
the greatest (resp.\ least)
solution of
a recursive equation $V = p \land \neXt\neXt V$
(resp.\
$V = p \lor \neXt\neXt V$).
To deal with both the least and greatest solutions,
$\muf$-calculus in \cite{JL07} classifies
each variable (i.e.\ the left-hand side of each equation)
into two types, $\mu$ and $\nu$.
A $\mu$-type (resp.\ $\nu$-type)
variable represents its least (resp.\ greatest) solution.
Moreover,
$\muf$-calculus has a syntactic restriction that
a $\mu$-type variable and a $\nu$-type one cannot
mutually depend on each other.
Similar restriction is imposed in
usual modal $\mu$-calculus~\cite{BW18}
that has $\mu$ and $\nu$ constructs
instead of systems of equations;
any two constructs cannot be mutually a subformula
of each other.
We do not follow these restrictions in $\mufo$-calculus,
because they prevent simple correspondence
between states of a \Buchi\ RA and
variables in $\mufo$-calculus equations.
Focusing on the fact that we only
need the same expressive power as \Buchi\ RA,
we define the semantics carefully and
keep the syntax and semantics of $\mufo$-calculus simple
without introducing restrictions
on the dependence between variables.

\section{Preliminaries}
\label{sec:preliminaries}

Let $\Nat$ be the set of natural numbers not including zero.
Let $[n] = \{1,\ldots,n\}$ for $n\in\Nat$.
$\powerset{X}$ denotes the power set of a set $X$.
$X^\omega$ is the set of infinite sequences over $X$.

Let $\At$ be a finite set of \emph{atomic propositions}.
Let $\Sigma = \powerset{\At}$ and we call it an \emph{alphabet}.
We assume a countable set $\Dat$ of \emph{data values}.
A sequence $w\in(\Sigma\times\Dat)^\omega$ is called
an infinite \emph{data word}.

For an infinite data word $w=(b_1,d_1)(b_2,d_2)\ldots$
where $b_i\in\Sigma$ and $d_i\in\Dat$ for $i\in\Nat$,
we let $w_i = (b_i,d_i)$ and
$\Rangew{i}{} = (b_i,d_i)(b_{i+1},d_{i+1})\ldots$
for $i\in\Nat$.
We also let
$\Str{(b,d)} = b$ and
$\Snd{(b,d)} = d$ for $b\in\Sigma$ and $d\in\Dat$.
Let $\bot\in\Dat$ be a data value designated as the initial value
of registers.

An \emph{assignment} of data values to $k$ registers is a $k$-tuple
$\theta\in D^k$.
The value of $r$-th register in an assignment $\theta$
is denoted by $\theta(r)$.
For $R\subseteq [k]$ and $d\in\Dat$,
$\theta[R\Gets d]$ is the assignment obtained from $\theta$
by updating the value of the $r$-th register for every $r\in R$
to~$d$.
Namely, $\theta[R\Gets d](r) = d$ for $r\in R$ and
$\theta[R\Gets d](r) = \theta(r)$ for $r\notin R$.
Let $\bot^k$ be the assignment that
assigns $\bot$ to all the $k$ registers.

\section{Disjunctive $\mu$-calculus with the freeze quantifier
($\mufo$-calculus)}

In the proposed logic named $\mufo$-calculus,
each specification is given as a system of equations
and each equation is given by a formula in
the disjunctive LTL with the freeze quantifier
($\LTLfo$) defined as follows.

\begin{definition}
\label{def:LTLfo}
\emph{$\LTLfo$ formulas} over
a set $\Var$ of variables, a set $\At$ of atomic propositions,
and $k$ registers
are defined by the following $\psi$,
where
$V\in\Var$, $p\in\At$,
$R\subseteq[k]$ and $r\in[k]$.
\begin{align*}
  \psi ::={} & V \mid \psi\lor\psi \mid
    {\Down_R}\neXt\psi\land\phi \mid \TT\,, \\
  \phi ::={} & \TT \mid \FF \mid p \mid \neg p \mid \Up_r \mid \neg\Up_r
    \mid \phi\land\phi\,.
\end{align*}
\emph{Basic $\LTLfo$ formulas} are
formulas defined by the above $\phi$.
Let $\Phi_k$ denote the set of basic $\LTLfo$ formulas.
\end{definition}

Temporal operator $\neXt$ (next) has the same meaning
as in the usual LTL\@.
The freeze operator $\Down_R$ represents the storing of
an input data value into registers specified by~$R$.
The look-up operator $\Up_r$ represents the proposition
that the stored value of the $r$-th register equals
the input data value.
Formal semantics of
these operators will be given later.
Since the class of RA is not closed under complement
and an RA does not have the ability to simulate logical conjunction,
negation can be applied to only atomic basic formulas,
and for a conjunction $\psi_1 \land \psi_2$,
at least one of $\psi_1$ and $\psi_2$ should
be a basic $\LTLfo$ formula.
Moreover, while the freeze operator $\Down_R$ can
simultaneously update any number of registers
with the same data value,
the operator must be followed by~$\neXt$.
We can consider $\neXt \psi$ is an $\LTLfo$ formula
since it is equivalent to $\Down_{\emptyset}\neXt \psi\land\TT$.
Similarly,
we can consider $\phi\in\Phi_k$ is
also an $\LTLfo$ formula
since $\phi$ is equivalent to $\Down_{\emptyset}\neXt\TT\land\phi$.

A basic $\LTLfo$ formula describes a proposition
on a single position in an infinite data word.
The semantics for a basic $\LTLfo$ formula
$\phi\in\Phi_k$
is provided by the relation
$w,i,\theta \models \phi$,
which intuitively means that
a position $i$ of an infinite data word $w$
satisfies $\phi$ when the contents of registers
equal~$\theta$.

\begin{definition}
For an infinite data word $w$, a position $i\in\Nat$,
and an assignment $\theta$,
\begin{alignat*}{2}
   w,i,\theta &\models p \quad
     &&{}:\Leftrightarrow
     p\in\Str{w_i}, \\
   w,i,\theta &\models {\Up_r}
     &&{}:\Leftrightarrow
     \theta(r) = \Snd{w_i}, \\
   w,i,\theta &\models \TT
     && \;\text{always holds},
\end{alignat*}
where $p\in\At$ and $r\in[k]$.
The semantics for negation and conjunction are the same as usual.
\end{definition}

\begin{definition}
\label{def:sys-of-eq}
A \emph{system of equations} (or a \emph{system})
$\sigma$ over
a set $\Var=\{V_1,\ldots,V_t\}$ of $t$ variables
is a mapping from
$\Var$ to $\LTLfo$ formulas over $\Var$.
A system $\sigma$ intuitively represents
a set of equations
$V_1=\sigma(V_1)$, $\ldots$\,, $V_t=\sigma(V_t)$.
For each system $\sigma$,
one variable is designated as its \emph{main variable}.
Moreover,
one or more variables are designated as
\emph{\omegavariable{}s},
and the set $\Varo$ $(\subseteq\Var)$ of
\omegavariable{}s is provided.
We assume that
$\Varo$ contains an \omegavariable\ $\Vtt$
such that $\sigma(\Vtt)=\TT$.
Moreover,
we assume that
for every $V\in\Varo$ and $V'\in\Var$,
$\sigma(V)=\sigma(V')$ implies
$V=V'$.\footnote{%
  This assumption causes no loss of generality because
  one can change
  $\sigma(V')$ into $\sigma(V')\lor(\FF\land\neXt\TT)$
  without altering the meanings.
}
\end{definition}

Intuitively, an $\LTLfo$ formula containing a variable $V$ equals
the $\LTLfo$ formula obtained by replacing $V$ with $\sigma(V)$.
For a recursively defined variable $V$,
if $V\in\Varo$,
then we consider $V$ intuitively equals
the formula obtained by replacing $V$ with $\sigma(V)$
\emph{infinitely many} times.
If $V\notin\Varo$,
then $V$ intuitively equals
the formula obtained by replacing $V$ with $\sigma(V)$
arbitrary finite times.
For example,
a formula $\Global\phi$ using
the global operator $\Global$ in LTL,
which means
``every position of a given infinite word satisfies $\phi$,''
can be simulated by an equation
$V=\phi\land\neXt V$ where $V\in\Varo$.
On the other hand,
a formula $\phi\Until\psi$
using the until operator $\Until$ in LTL,
which means
``some position of a given infinite word satisfies $\psi$ and
every position before that position satisfies $\phi$,''
can be simulated by an equation
$V=\psi\lor(\phi\land\neXt V)$ where $V\notin\Varo$.

Formally,
we introduce
a 5-tuple $(i,\theta; j,\theta',x)$
where $i,j\in\Nat$ are positions in an infinite word,
$\theta,\theta'$ are assignments,
and $x$ is an \omegavariable,
and let the ``value'' of each variable $V\in\Var$
be a set of these 5-tuples.
Intuitively, the value of a variable $V$ contains
$(i,\theta;j,\theta',x)$ if and only if
$(w,i,\theta)$ satisfies $V$
under the assumption that
$(w,j,\theta')$ satisfies $x$,
where $w$ is a given infinite word.
The sentence
``$(w,i,\theta)$ satisfies $V$'' intuitively means
that $\sigma(V)$ holds at position $i$ in word $w$
when the contents of registers equal~$\theta$.
If we were considering only the least solution
of a system of equations,
then it was sufficient to let the value of a variable
be a set of pair $(i,\theta)$
as the same as $\mufd$-calculus in~\cite{TOSS23}.
Instead of considering an infinite substitution
of $\sigma(V)$ for an \omegavariable\ $V$,
we consider an infinite sequence $i_1,i_2,\ldots$ of
positions each of which satisfies~$V$.
To handle such a sequence, we use 5-tuple
$(i,\theta;j,\theta',V)$ instead of pair $(i,\theta)$,
where $(j,\theta')$ represents ``the next position''
of $(i,\theta)$ satisfying $V$
(and the contents of registers at that position).
The tuple asserts that
whether $(w,i,\theta)$ satisfies $V$
depends on whether
$(w,j,\theta')$ satisfies $V$.
The latter may also depend on
whether some
$(w,j_1,\theta_1)$ satisfies $V$.
We consider $(w,i,\theta)$ satisfies $V$
if such an infinite chain of dependence exists.

An \emph{environment} is the mapping from a variable to
the above-mentioned set of 5-tuples.
\begin{definition}
An \emph{environment}
for a given infinite data word $w\in(\Sigma\times\Dat)^\omega$
is a mapping of type
$\Var\to\powerset{\Nat \times \Dat^k \times \Nat \times \Dat^k
 \times \Varo}$.
Let $\Env_w$ be the set of environments for $w$.
\end{definition}

\begin{definition}
\label{def:ltlfo-semantics}
The following relation gives
the semantics for each variable in
$\LTLfo$ formulas over a given environment.
For an infinite data word $w$, a position $i\in\Nat$,
an assignment $\theta$, an environment $u\in\Env_w$,
and a variable $V\in\Var$,
\begin{alignat*}{2}
   w,i,\theta &\models_{u} V
     &{}:\Leftrightarrow{}&
     \exists j,\theta',x\colon
     (i,\theta;j,\theta',x)\in u(V), \\ &&&\ 
     i < j
     \text{ and }
     w,j,\theta'\models_{u} x.
\end{alignat*}
\end{definition}

A system $\sigma$ of equations specifies
an environment
under which all equations $V=\sigma(V)$ hold.
Such an environment can be seen as a \emph{solution} of
the system of equations.
We define an augmented version of
the semantic relation $\models_u$ as follows,
and let a solution of $\sigma$ be
an environment $u$ satisfying
$w,i,\theta,j,\theta',x\models_u V \iff
 w,i,\theta,j,\theta',x\models_u \sigma(V)$.
\begin{definition}
\label{def:ltlfo-semantics-fin}
For an infinite data word $w$, positions $i,j\in\Nat$ such that $i\le j$,
assignments $\theta,\theta'$, $x\in\Varo$,
and an environment $u\in\Env_w$,
\begin{alignat*}{2}
   w,i,\theta,j,\theta',x &\models_{u} V
     &{}:\Leftrightarrow{}&
     (i,\theta;j,\theta',x)\in u(V), \\
   w,i,\theta,j,\theta',x &\models_{u} \psi_1\lor\psi_2
     &{}:\Leftrightarrow{}&
     (w,i,\theta,j,\theta',x \models_{u} \psi_1) \\ && \text{or }&
     (w,i,\theta,j,\theta',x \models_{u} \psi_2), \\
   w,i,\theta,j,\theta',x &\models_{u} {\Down_R}\neXt\psi\land\phi \
     &{}:\Leftrightarrow{}&
     i < j \text{ and }
     w,i,\theta\models\phi \\
     \text{and}
     & \rlap{ $(w,i+1,\theta[R\Gets\Snd{w_i}],j,\theta',x \models_{u} \psi)$,}
     \\
   w,i,\theta,j,\theta',x &\models_{u} \TT \
     &{}:\Leftrightarrow{}&
     \theta = \theta' \text{ and } x=\Vtt,
\end{alignat*}
where $V\in\Var$, $R\subseteq[k]$, and $\phi\in\Phi_k$.
\end{definition}

The (least) solution of $\sigma$
can be represented as
the \emph{least fixed point} of
the variable updating function $\Fsigmaw$
in Definition~\ref{def:mapping-F}.
We define a partial order $\subseteq$ over $\Env_w$ as:
\begin{equation*}
    u_1 \subseteq u_2 \; :\Leftrightarrow \;
    u_1(V)\subseteq u_2(V) \text{ for every } V\in\Var.
\end{equation*}
Let $\emptyset$ denote the environment
that maps every $V\in\Var$ to $\emptyset$.
Obviously, $\emptyset\subseteq u$ for any environment~$u$.
The least upper bound $u_1\cup u_2$ of $u_1,u_2\in\Env_w$
can be represented as
$(u_1\cup u_2)(V) = u_1(V)\cup u_2(V)$ for every $V\in\Var$.
Then,
$(\Env_w,{\subseteq})$ is
a complete partial order (CPO)\footnote{%
  A partial order $(D,{\le})$ is a 
  CPO if every directed subset $S\subseteq D$
  has its least upper bound~$\bigvee S$.
  A set $S$ is \emph{directed} if for every finite subset
  $S_0\subseteq_{\mathrm{fin}} S$, an upper bound of $S_0$
  exists in~$S$.
  $(\Env_w,{\subseteq})$ is a CPO because
  every (not necessarily directed) subset $S\subseteq\Env_w$
  has its least upper bound $\bigcup S$.
}
\cite[Sect.~5.2]{Mit96}
with the least element $\emptyset$.

\begin{definition}
\label{def:mapping-F}
For a system $\sigma$ of equations over $\Var$
and an infinite data word $w \in (\Sigma \times\Dat)^\omega$,
we define a \emph{variable updating mapping}
$\Fsigmaw:\Env_w\to\Env_w$ as:
\begin{align*}
  & \Fsigmaw(u)(V) := \\ & \
  \begin{cases}
    \begin{aligned}
    \{(i,\theta;j,\theta',x) \mid
      w,i,\theta,j,\theta',x \models_{u} \sigma(V) \} \\
    {}\cup
      \{(i,\theta;i,\theta,V)\mid i\in\Nat,\ \theta\in\Dat^k\}
    \end{aligned}
      & \text{if } V\in\Varo, \\
    \{(i,\theta;j,\theta',x) \mid
      w,i,\theta,j,\theta',x \models_{u} \sigma(V) \}
      & \text{otherwise},
  \end{cases}
\end{align*}
for each $u\in\Env_w$ and $V\in\Var$.
\end{definition}

\begin{definition}
  For a variable updating mapping $F$,
  a \emph{fixed point} of $F$ is an environment $u$ that satisfies
  $F(u)=u$.
  A fixed point $u$ of $F$ is called
  the \emph{least fixed point} of $F$,
  denoted by $\lfp(F)$,
  if $u\subseteq u'$ for any fixed point $u'$ of $F$.
\end{definition}

We can prove that
$\Fsigmaw$ is \emph{continuous},%
\footnote{%
  A function $f:D\to D$ over
  a CPO $(D,{\le})$ is \emph{continuous} if
  $f(\bigvee S)=\bigvee_{x\in S}f(x)$ for
  every directed subset $S\subseteq D$.
  We can show
  $\Fsigmaw$ satisfies
  $\Fsigmaw(\bigcup S) = \bigcup_{u\in S}\Fsigmaw(u)$
  for every non-empty (not necessarily directed) $S\subseteq\Env_w$.
}
which implies that
$\Fsigmaw$ has the least fixed point satisfying
$\lfp(\Fsigmaw)=\bigcup_{n=0}^{\infty}\Fsigmaw^n(\emptyset)$.

\begin{definition}
  For an infinite data word $w$ and a system $\sigma$
  with its main variable $V_t$,
  we say $w$ satisfies $\sigma$, denoted by $w\models\sigma$,
  if $w,1,\bot^k \models_{\lfp(\Fsigmaw)} V_t$,
  where $\models_{\lfp(\Fsigmaw)}$ is the relation defined in
  Definition~\ref{def:ltlfo-semantics}.
\end{definition}

\begin{example}\label{example:sigma1}
Consider a system $\sigma_1$ of equations
\begin{equation*}
\sigma_1=\{\,
\begin{aligned}[t]
   &\Vtt=\TT,\ V_1=\Up_1,\\
   &V_2=V_1\lor(\neXt V_2\land(\neg\Up_1\land p_1)),\
   V_3=\Down_1\neXt V_2
   \,\}
\end{aligned}
\end{equation*}
over $\Var=\{\Vtt,V_1,V_2,V_3\}$ with
$\Varo=\{\Vtt\}$.
The main variable is $V_3$.
This system of equations
is intuitively equivalent to
$\Down_1\neXt((\neg\Up_1\land p_1)\Until\Up_1)$
where $\Until$ is the until operator in LTL\@,
and an infinite data word
$(b_1,d_1)(b_2,d_2)\ldots\in(\Sigma\times\Dat)^\omega$ satisfies
$\sigma_1$ if and only if
$d_1=d_n$ for $\exists n>1$
and $p_1\in b_i$ and $d_1\neq d_i$ for
$\forall i\in\{2,\ldots,n-1\}$.

Let $u^\ell=\Fsigmaiw^\ell(\emptyset)$.
By definition, $u^\ell$ for $\ell\ge1$
satisfies the following equation:%
\footnote{%
  Remember that $V_1=\Up_1$ is an abbreviation
  of $V_1=\Down_{\emptyset}\neXt\TT\land\Up_1$.
}
\begin{align*}
  u^\ell ={} &\{\,
    \Vtt \mapsto \{ (i,\theta; j,\theta,\Vtt) \mid i\le j \}, \\
    V_1 &\mapsto \{ (i,\theta; j,\theta,\Vtt) \mid i < j,\
      \theta(1)=\Snd{w_i} \}, \\
    V_2 &\mapsto u^{\ell-1}(V_1) \cup{}
         \{ (i,\theta; j,\theta',x) \mid
          i < j,\ \theta(1)\ne\Snd{w_i},\\
        &\qquad 
          p_1\in\Str{w_i},\
          (i+1,\theta; j,\theta',x) \in u^{\ell-1}(V_2) \}, \\
    V_3 &\mapsto \{ (i,\theta; j,\theta',x) \mid i < j,\\
        &\qquad
          (i+1,\theta[\{1\}\Gets\Snd{w_i}]; j,\theta',x)
          \in u^{\ell-1}(V_2) \}
  \,\}.
\end{align*}
Since $\sigma_1(\Vtt)$ and $\sigma_1(V_1)$
contain no variable,
$u^\ell(\Vtt)=u^1(\Vtt)$ and $u^\ell(V_1)=u^1(V_1)$
for any $\ell\ge1$.
Assume that $w=(\emptyset,5)\allowbreak(\{p_1,p_2\},4)\allowbreak
(\{p_1\},4)\allowbreak(\{p_1\},5)\allowbreak\ldots$
and $w_i=w_4$ for all $i\ge5$.
Let $[d]$ denote the assignment where the value of
register $1$ is~$d$. Then,
the following equations hold:
\begin{align*}
  u^1(V_1)&=\{\,
    \begin{aligned}[t]
    &(1,[5];2,[5],\Vtt), (1,[5];3,[5],\Vtt),\ldots,\\
    &(2,[4];3,[4],\Vtt), (2,[4];4,[4],\Vtt),\ldots, \\
    &(3,[4];4,[4],\Vtt), (3,[4];5,[4],\Vtt),\ldots, \\
    &(4,[5];5,[5],\Vtt), (4,[5];6,[5],\Vtt),\ldots, \\
    &(5,[5];6,[5],\Vtt), (5,[5];7,[5],\Vtt),\ldots\,
  \},
    \end{aligned} \\[\smallskipamount]
  u^1(V_2)&=\emptyset, \quad
  u^2(V_2)=u^1(V_1), \\
  u^3(V_2)&=u^2(V_2)\cup{}\\&\hphantom{{}={}}
  \{\,(3,[5];5,[5],\Vtt), (3,[5];6,[5],\Vtt),\ldots\,\},\\
  u^4(V_2)&=u^3(V_2)\cup{}\\&\hphantom{{}={}}
  \{\,(2,[5];5,[5],\Vtt), (2,[5];6,[5],\Vtt),\ldots\,\},\\
  u^5(V_2)&=u^4(V_2),
  \\[\smallskipamount]
  u^1(V_3)&=u^2(V_3)=\emptyset,\\
  u^3(V_3)&=\{ (2,[d];j,[4],\Vtt)\mid d\in\Dat,\ j\ge4 \}
    \cup{}\\ &\hphantom{{}={}}
    \{ (\rlap{$i,$}\hphantom{2,{}}
      [d];j,[5],\Vtt)\mid d\in\Dat,\
      4\le i,\ i+2\le j \},\\
  u^4(V_3)&=u^3(V_3),\\
  u^5(V_3)&=u^4(V_3)\cup
   \{ (1,[d];j,[5],\Vtt)\mid d\in\Dat,\ j\ge5 \}.
\end{align*}
For $\ell\ge6$, $u^\ell=u^5$. Thus,
$\lfp(\Fsigmaiw)=u^5$.
Since
$(i,[5];j,\allowbreak{}[5],\Vtt)\in u^5(\Vtt)$ for any $i,j$
with $5\le i< j$, we have
$w,5,[5]\models_{u^5}\Vtt$.
Therefore,
$w,1,[\bot]\models_{u^5} V_3$ and thus $w\models\sigma_1$.
\end{example}

\begin{example}\label{example:sigma2}
Let $\sigma_2$ be
the same system of equations
as $\sigma_1$ in Example~\ref{example:sigma1}
except that $\Varo=\{\Vtt,V_2\}$,
which is
intuitively equivalent to
$\Down_1\neXt((\neg\Up_1\land p_1)\WeakUntil\Up_1)$
where $\WeakUntil$ is the weak until operator.\footnote{%
$\phi\WeakUntil\psi \equiv (\phi\Until\psi)\lor(\Global\phi)$.
}
Every infinite data word that satisfies $\sigma_1$ also
satisfies $\sigma_2$.
Moreover, an infinite data word
$(b_1,d_1)(b_2,d_2)\ldots$ such that
$p_1\in b_i$ and $d_1\neq d_i$ for $\forall i\ge 2$
also satisfies~$\sigma_2$.

Consider $w'=(\emptyset,3)\,\Rangew{2}{}$ where
$w$ is the data word in Example~\ref{example:sigma1}.
Let $u_i^\ell = \FsigmaIwi^\ell(\emptyset)$
for $i\in\{1,2\}$ and
$u^\ell$ be the environment defined
in Example~\ref{example:sigma1}.
Then, we have
$u_1^\ell(V_3)=u_1^3(V_3)=u^3(V_3)$ for $\ell\ge4$.
Thus, $w'\!,1,\theta\not\models_{u_1^\ell} V_3$
for any $\theta$ and $\ell$.
On the other hand,
$u_2^1(V_2)=\{(i,[d];i,[d],V_2)\mid d\in\Dat,\ i\in\Nat\}$
since $V_2\in\Varo$,
which results in
$(i,[3];i+1,[3],V_2)\in u_2^2(V_2)$ for all $i\ge2$
and
$(1,[d];2,[3],V_2)\in u_2^2(V_3)$ for all $d\in\Dat$.
Therefore,
$w'\!,2,[3]\models_{u_2^2} V_2$ and
$w'\!,1,[\bot]\models_{u_2^2} V_3$,
and thus $w'\models\sigma_2$.
\end{example}

The following lemma can be proved in the same way
as in $\mufd$-calculus~\cite[Lemma~3]{TOSS23}.
In later sections,
we consider only $\LTLfo$ formulas in a normal form
stated in this lemma.
\begin{lemma}
\label{lemma:normalform}
  Let $\sigma$ be a system of equations in $\mufo$-calculus.
  We can translate $\sigma$
  into an equivalent system
  $\sigma'$ that consists of equations of the form
  $V=V'\lor V''$, $V=\Down_R\neXt V'\land\phi$, or $V=\TT$
  where $V,V',V''\in\Var$ and $\phi\in\Phi_k$,
  by applying the following translations (i) and (ii) iteratively.
  The newly introduced variables $V'$ and $V''$
  in (i) and (ii)
  should not be \omegavariable{}s.
\begin{enumerate}\itemsep=0pt
\renewcommand{\labelenumi}{(\roman{enumi})}
\item
  Replace $V=\psi_1\lor\psi_2$ with
  $V' = \psi_1$, $V'' = \psi_2$, and $V=V' \lor V''$.
\item
  Replace $V=\Down_R\neXt\psi_1 \land\phi$ with
  $V' = \psi_1$ and $V=\Down_R \neXt V' \land\phi$.
\end{enumerate}
\end{lemma}

\section{\Buchi\ register automata}
\label{sec:BRA}

\begin{definition}
A \emph{\Buchi\ register automaton with $k$ registers} ($k$-BRA)
over $\Sigma$ and $\Dat$ is a quadruple
$\calA = (Q,{q_0},\delta,F)$, where
$Q$ is a finite set of \emph{states},
$q_0\in Q$ is the \emph{initial state},
$\delta \subseteq Q \times (\Phi_{k}\cup\{\varepsilon\})
  \times \powerset{[k]} \times Q$ is a set of \emph{transition rules},
  and
$F\subseteq Q$ is a set of \emph{accepting states}.
A transition rule $(q,\phi,R,q')\in\delta$ is written as
$(q,\phi,R)\TO q'$ for readability.
A transition rule whose second element is $\varepsilon$ is called
an $\varepsilon$-rule.
For each $\varepsilon$-rule $(q,\varepsilon,R)\TO q'\in\delta$,
$R$ must be $\emptyset$.
\end{definition}

The second element $\phi$ of a transition rule is called
the \emph{guard condition} of that rule,
which is either an element of $\Phi_k$ or $\varepsilon$.
An expression $\phi\in\Phi_k$ represents the condition
under which
a transition rule can be applied to an element of a data word.
For example,
a transition rule with guard condition $p_1\land\Up_2$ can be
applied to the current position $i$ in an input data word $w$
if $p_1\in\Str{w_i}$ and $\Snd{w_i}$ equals
the contents of the second register.
An $\varepsilon$-rule represents a transition without
consuming the input.

The third element $R\subseteq[k]$ of a transition rule designates
registers to be updated.
$R$ can be empty, which represents no register is updated.

For a $k$-BRA $\calA = (Q,q_0,\delta,F)$, every triple
$(q,\theta,w')\in Q\times\Dat^k\times (\Sigma\times\Dat)^\omega$ is called
an \emph{instantaneous description} (ID) of $\calA$,
where
$q$, $\theta$, and $w'$ represent
the current state, the current assignment, and the remaining suffix
of the input data word, respectively.
We define the transition relation $\vdash_{\calA}$ over IDs
by the following inference rules:
\begin{align*}
&\begin{array}{l}
  (q,\phi,R)\TO q' \in\delta \quad
  w,i,\theta \models \phi \\
  \hline
  (q,\theta,\Rangew{i}{}) \vdash_{\calA} (q',\theta[R\Gets\Snd{w_i}],\Rangew{i+1}{})
\end{array} \\[\smallskipamount]
&\begin{array}{l}
  (q,\varepsilon,\emptyset)\TO q' \in \delta \quad
  i\in\Nat \\
  \hline
  (q,\theta,\Rangew{i}{}) \vdash_{\calA} (q',\theta,\Rangew{i}{})
\end{array}
\end{align*}
Note that
we can eliminate $\varepsilon$-rules in the same way as
usual non-deterministic finite automata,
because $\varepsilon$-rules in our definition cannot
impose any constraints on data values or update any registers.

\begin{example}
Consider a $2$-BRA over $\At=\{p_1,p_2,p_3\}$
that has a transition rule
$(q,p_1\land\neg p_3\land\Up_1,\{2\})\TO q'$.
Let $w=(\{p_1\},5)(\{p_1,p_3\},5)(\{p_1\},4)\ldots$ and
$\theta\in\Dat^2$ be an assignment with $\theta(1)=5$.
Then,
$(q,\theta,\Rangew{1}{})\vdash_{\calA}(q',\theta[\{2\}\Gets5],\Rangew{2}{})$,
but
the transition rule cannot be applied to
$(q,\theta,\Rangew{2}{})$ or $(q,\theta,\Rangew{3}{})$.
\end{example}

We say that $\calA$ \emph{accepts} $w$ if
there exist $q_F\in F$, $i_1,i_2,\ldots\in\Nat$, and
$\theta_1,\theta_2,\ldots\in\Dat^k$ such that
$i_1<i_2<\cdots$ and
$(q_0,\bot^k,w) \vdash_{\calA}^* (q_F,\theta_1,\Rangew{i_1}{})
\vdash_{\calA}^* (q_F,\theta_2,\Rangew{i_2}{})
\vdash_{\calA}^* \cdots$.
The data language \emph{recognized} by $\calA$ is defined as
$L(\calA) = \{ w \mid\calA$ accepts $w\}$.

\section{Transformation from $\mufo$-calculus into \Buchi\ RA}
\label{sec:muf-to-ra}

\begin{theorem}
  \label{th:muf-to-ra}
  Let $\sigma$ be a system of equations 
  over $\Var$, 
  whose main variable is $V_t$,
  in the normal form described in Lemma~\ref{lemma:normalform}.
  The following \Buchi\ RA $\calA = (Q,\sigma(V_t),\delta,F)$
  is equivalent to $\sigma$\,;
  i.e., $L(\calA)=\{w\mid w\models\sigma\}$.
\begin{align*}
  Q &= 
       \{\sigma(V_i) \mid V_i\in\Var\}, \\[\smallskipamount]
  \delta &= \{ (\TT,\TT,\emptyset)\TO\TT \} \\
  &\cup \{(V_i\lor V_j,\varepsilon,\emptyset)\TO\sigma(V_i) \mid
           V_i\lor V_j \in Q \} \\
  &\cup \{(V_i\lor V_j,\varepsilon,\emptyset)\TO\sigma(V_j) \mid
           V_i\lor V_j \in Q \} \\
  &\cup \{(\Down_R\neXt V_i \land \phi, \phi, R) \TO \sigma(V_i) \mid
           \Down_R\neXt V_i \land \phi \in Q \},
  \\[\smallskipamount]
  F &= 
       \{\sigma(\Vo)\mid\Vo\in\Varo\}.
\end{align*}
\end{theorem}

Theorem~\ref{th:muf-to-ra} is a direct consequence of
the following Lemmas~\ref{lemma:mu-to-ra-fin}
and \ref{lemma:mu-to-ra}:

\begin{lemma}
\label{lemma:mu-to-ra-fin}
  $w,i,\theta,j,\theta',x \models_{\lfp(\Fsigmaw)} V$
  iff
  $(\sigma(V),\theta,\Rangew{i}{}) \vdash_{\calA}^*
  (\sigma(x),\theta',\Rangew{j}{})$
  and
  $\sigma(x)\in F$.
\end{lemma}
\begin{proof}
Let $u^\ell = \Fsigmaw^\ell(\emptyset)$.

$(\Rightarrow)$
We prove that
$w,i,\theta,j,\theta',x \models_{u^\ell} V$
implies
the right-hand side
by induction on $\ell$.
When $\ell=0$, it vacuously holds since
$u^0=\emptyset$ and
$w,i,\theta,j,\theta',x \models_{\emptyset} V$
never holds.
Assume that
$w,i,\theta,j,\theta',x \models_{u^{\ell+1}} V$
for some
$\ell,V,i,j,\theta,\theta',x$ where $\ell\ge0$ and $i\le j$.
Since $x\in\Varo$,
$\sigma(x)\in F$ obviously holds.
By the definition of $\models_{u^{\ell+1}}$,
$(i,\theta; j,\theta',x)\in u^{\ell+1}(V)$.
By the definition of $\Fsigmaw$, there are two cases:
$w,i,\theta,j,\theta',x \models_{u^\ell} \sigma(V)$,
or $V\in\Varo$ and $(i,\theta,V)=(j,\theta',x)$.
In the latter case,
$(\sigma(V),\theta,\Rangew{i}{}) \vdash_{\calA}^*
(\sigma(x),\theta',\Rangew{j}{})$ 
trivially holds.
We consider the former case and conduct a case analysis
on $\sigma(V)$.
By Lemma~\ref{lemma:normalform}, there are three cases:

({i}) \ When $\sigma(V)=V_1\lor V_2$,
by the definition of $\models_{u^\ell}$,
either
$w,i,\theta,j,\theta',x \models_{u^\ell} V_1$ or
$w,i,\theta,j,\theta',x \models_{u^\ell} V_2$ holds.
In either case
$(V_1\lor V_2,\theta,\Rangew{i}{}) \vdash_{\calA}^*
(\sigma(x),\theta',\Rangew{j}{})$ holds
by the induction hypothesis
and the fact that $\calA$ has transition rules
$(V_1\lor V_2, \varepsilon, \emptyset)\to \sigma(V_1)$ and
$(V_1\lor V_2, \varepsilon, \emptyset)\to \sigma(V_2)$.

({ii}) \ When $\sigma(V)=\Down_R\neXt V_1\land\phi$,
by the definition of $\models_{u^\ell}$,
$i<j$ and $w,i,\theta\models\phi$ and
$(w,i+1,\theta[R\Gets\Snd{w_i}],\allowbreak
j,\theta',x\allowbreak \models_{u^\ell} V_1)$ hold.
By the definition of $\calA$ and the induction hypothesis,
$(\Down_R\neXt V_1\land\phi,\theta,\Rangew{i}{})
\vdash_{\calA} (\sigma(V_1),\theta[R\Gets\Snd{w_i}],\allowbreak
\Rangew{i+1}{})\allowbreak
\vdash_{\calA}^*
(\sigma(x),\theta',\Rangew{j}{})$. 

({iii}) \ When $\sigma(V)=\TT$,
by the definition of $\models_{u^\ell}$,
$\theta=\theta'$ and $x=\Vtt$ hold.
Since $\calA$ has transition rule $(\TT,\TT,\emptyset)\to\TT$,
$(\TT,\theta,\Rangew{i}{}) \vdash_{\calA}^*
(\TT,\theta,\Rangew{j}{})$ for any $i$ and $j$ with $i\le j$.

$(\Leftarrow)$
The backward direction can be proved by showing
the right-hand side implies $\exists \ell\colon
\left(w,i,\theta,j,\theta',x \models_{u^\ell} V\right)$
by induction on the number of state transitions.
By $\sigma(x)\in F$, $x\in\Varo$ holds.

(The base case)
Assume that
$(\sigma(V),\theta,\Rangew{i}{})=
 (\sigma(x),\allowbreak\theta',\allowbreak\Rangew{j}{})$.
Since $x\in\Varo$, $x=V$ holds because we have assumed that
$\sigma(V_1)=\sigma(V_2)$ implies $V_1=V_2$ for every
$V_1\in\Varo$.
By the definition of $\Fsigmaw$,
$(i,\theta;i,\theta,V)\in u^1(V)$ holds,
and thus
$w,i,\theta,i,\theta,x \models_{u^1} V$.

(The induction step)
Assume that
$(\sigma(V),\theta,\Rangew{i}{})\vdash_{\calA}
 (q_1,\theta_1,\Rangew{i_1}{})\vdash_{\calA}^*
 (\sigma(x),\theta',\Rangew{j}{})$
where $i\le i_1\le j$.
We conduct a case analysis on $\sigma(V)$.

({i}) \ When $\sigma(V)=V_1\lor V_2$,
by the definition of $\calA$,
$(\theta_1,i_1)=(\theta,i)$ and $q_1$ should equal
either $\sigma(V_1)$ or $\sigma(V_2)$.
When $q_1=\sigma(V_1)$ (resp.\ $\sigma(V_2)$),
by the induction hypothesis,
there exists some $\ell$ that satisfies
$w,i,\theta,j,\theta',x \models_{u^{\ell}} V_1$
(resp.\ $V_2$).
Therefore in either case
$w,i,\theta,j,\theta',x \models_{u^{\ell}} V_1\lor V_2$
holds,
which implies that
$w,i,\theta,j,\theta',x \models_{u^{\ell+1}} V$
regardless of whether $V\in\Varo$ or not.

({ii}) \ When $\sigma(V)=\Down_R\neXt V_1\land\phi$,
by the definition of $\calA$,
$(q_1,\theta_1,i_1)=
 (\sigma(V_1),\theta[R\Gets\Snd{w_i}],i+1)$
and $w,i,\theta\models\phi$.
By the induction hypothesis,
there exists some $\ell$ that satisfies
$\left(w,i+1,\theta[R\Gets\Snd{w_i}],j,\theta',x
 \models_{u^{\ell}} V_1\right)$,
which implies
$w,i,\theta,j,\theta',x
 \models_{u^{\ell}} \sigma(V)$ and thus
$w,i,\theta,j,\theta',x
 \models_{u^{\ell+1}} V$.

({iii}) \ When $\sigma(V)=\TT$,
by the definition of $\calA$,
$(q_1,\theta_1,i_1)=(\TT,\theta,i+1)$.
Since $(\TT,\TT,\emptyset)\to \TT$
is the only transition rule of $\calA$
that can be applied to an ID whose state is $\TT$,
$(\sigma(x),\theta')$ should equal $(\TT,\theta)$.
In the same way as the base case,
we can show $V=x=\Vtt$.
Therefore
$w,i,\theta,j,\theta,x\models_{\emptyset} \TT$ and
thus
$w,i,\theta,j,\theta,x\models_{u^1} V$.
\QED
\end{proof}

\begin{lemma}
\label{lemma:mu-to-ra}
  $w,i,\theta\models_{\lfp(\Fsigmaw)} V$
  iff
  $(\sigma(V),\theta,\Rangew{i}{}) \vdash_{\calA}^*
  (q_F,\theta_1,\Rangew{i_1}{}) \vdash_{\calA}^*
  (q_F,\theta_2,\Rangew{i_2}{}) \vdash_{\calA}^* \cdots$
  for some
  $q_F\in F$ and
  $i_1,\theta_1,i_2,\theta_2,\ldots$
  such that $i<i_1<i_2<\cdots$.
\end{lemma}
\begin{proof}
Let $u=\lfp(\Fsigmaw)$.

$(\Rightarrow)$
Since $w,i,\theta\models_u V$,
there exist
$i_1,\theta_1,x_1$ such that
$(i,\theta;i_1,\theta_1,x_1)\in u(V)$,
$i<i_1$, and
$w,i_1,\theta_1\models_u x_1$.
If $x_1\in\Varo$, then
there exist $i_2,\theta_2,x_2$ such that
$(i_1,\theta_1;i_2,\theta_2,x_2)\in u(x_1)$,
$i_1<i_2$, and
$w,i_2,\theta_2\models_u x_2$.
If $x_2\in\Varo$, then
there exist $i_3,\theta_3,x_3$ such that
$(i_2,\theta_2;i_3,\theta_3,x_3)\in u(x_2)$,
$i_2<i_3$, and
$w,i_3,\theta_3\models_u x_3$,
and so on.
By Lemma~\ref{lemma:mu-to-ra-fin},
$(\sigma(V),\theta,\Rangew{i}{}) \vdash_{\calA}^*
 (\sigma(x_1),\theta_1,\Rangew{i_1}{}) \vdash_{\calA}^*
 (\sigma(x_2),\theta_2,\Rangew{i_2}{}) \vdash_{\calA}^*
\cdots$.
There must exist
$\Vo\in\Varo$ that appears infinitely often in the sequence
$x_1,x_2,\ldots\,$,
and the statement holds by letting $q_F=\sigma(\Vo)$.

$(\Leftarrow)$
When $q_F=\sigma(\Vo)$ for some $\Vo\in\Varo$,
$(i,\theta;i_1,\theta_1,\Vo)\in u(V)$,
$(i_1,\theta_1;i_2,\theta_2,\Vo)\in u(\Vo)$,
\ldots\ 
by Lemma~\ref{lemma:mu-to-ra-fin}
because
$(\sigma(V),\theta,\Rangew{i}{}) \vdash_{\calA}^*
(\sigma(\Vo),\theta_1,\Rangew{i_1}{}) \vdash_{\calA}^*
(\sigma(\Vo),\theta_2,\Rangew{i_2}{}) \vdash_{\calA}^*
\cdots$.
Therefore
$w,i,\theta\models_u V$.
\QED
\end{proof}

\begin{example}
A \Buchi\ RA equivalent to $\sigma_1$
in Example~\ref{example:sigma1}
is shown in Figure~\ref{fig:muf-to-ra}.
By additionally designating
the state $V_1\lor V'_2$ $(=\sigma_2(V_2))$ as an accepting state
in Figure~\ref{fig:muf-to-ra},
we obtain
a \Buchi\ RA equivalent to $\sigma_2$ in Example~\ref{example:sigma2}.
\end{example}

\begin{figure}[tb]\centering
  \begin{tikzpicture}
    [x=0.9cm,y=0.9cm,
     state/.style={rounded rectangle,draw,text height=2.2ex,text depth=.9ex}]
    \node[initial,state,anchor=east] (q0) at (-1.8,0) {$\,\Down_1\neXt V_2\,$};
    \node[state,anchor=center] (q1) at (0,0) {$\,V_1\lor V'_2\,$};
    \node[state,anchor=west] (q2) at (2.5,0) {$\,\neXt V_2\land(\neg\Up_1\land p_1)\,$};
    \node[state,anchor=center] (q3) at (0,-1.5) {$\,\neXt\TT\land\Up_1\,$};
    \node[state,accepting,anchor=west] (q4) at (2.5,-1.5) {$\TT$};
    \path
      (q0) edge node {$\TT,\{1\}$} (q1)
      (q1) edge node [below] {$\varepsilon$} (q2)
      (q2) edge [out=174,in=10] node [above] {$\neg\Up_1\land p_1,\emptyset$} (q1)
      (q1) edge node {$\varepsilon$} (q3)
      (q3) edge node {$\Up_1,\emptyset$} (q4)
      (q4) edge [loop right] node {$\TT,\emptyset$} (q4);
  \end{tikzpicture}
  \caption{\Buchi\ RA equivalent to
      $\{\,\Vtt=\TT$, $V_1 = \Up_1$,
          $V_2 = V_1 \lor(\neXt V_2\land\allowbreak
          (\neg\Up_1\land p_1))$,
          $V_3 = \Down_1\neXt V_2 \,\}$
           when $\Varo=\{\Vtt\}$.}
  \label{fig:muf-to-ra}
\end{figure}

\section{Transformation from \Buchi\ RA into $\mufo$-calculus}
\label{sec:ra-to-muf}

\begin{theorem}\label{th:ra-to-muf}
Let $\calA = (Q,q_0,\delta,F)$ be a \Buchi\ RA
without $\varepsilon$-rules,
where every $q\in Q$ has at least one state transition from it.
We define the set $\Var$ of variables,
the set $\Varo$ of \omegavariable{}s, and
the system $\sigma$ of equations over $\Var$,
whose main variable is $V_{q_0}$, as follows.%
\footnote{%
  We let $\Var$ contain $\Vtt$ only for conforming
  to Definition~\ref{def:sys-of-eq}.
}
Then, $\sigma$ is equivalent to~$\calA$.
\begin{align*}
\Var &= \{V_q \mid q\in Q\} \cup \{V_r \mid r\in\delta\}
  \cup \{\Vtt\}, \\
\Varo &= \{V_q\mid q\in F\} \cup \{\Vtt\},
 \\[\smallskipamount]
 \sigma &
\begin{aligned}[t]
    {}={}   &\textstyle
             \{ V_q = \bigvee_{r\in\delta|_q}V_r \mid q\in Q\}\\
    {}\cup{}&\{ V_r=\Down_R\, \neXt V_{q'} \land \phi \mid r\in\delta,\
                r=(q,\phi,R) \TO q' \} \\
    {}\cup{}&\{ \Vtt = \TT \}, \\
  \end{aligned}
\end{align*}
where
  $\delta|_q = \{r\in\delta\mid r=(q,\phi,R) \TO q'$
     for $\exists\phi,R,q'\}$.
\end{theorem}
\begin{proof}
By applying Theorem~\ref{th:muf-to-ra} to $\sigma$
constructed from a given $\calA$ as above,
we obtain the following BRA
$\calA'=(Q',q'_0,\delta',F')$.%
\footnote{
  State $\TT$ of $\calA'$,
  which is brought by $\Vtt$, has no effect
  because it is not reachable from the initial state.
}
(In the construction of $\calA'$,
we use an extended version of the construction in Theorem~\ref{th:muf-to-ra}
where a disjunction $V_{j_1}\lor\cdots\lor V_{j_m}$
of any number of variables
is allowed as the right-hand side of an equation,
and transition rule
$(V_{j_1}\lor\cdots\lor V_{j_m},\varepsilon,\emptyset)\TO\sigma(V_{j_i})$
for each $i\in[m]$ is constructed for it.
This extension does not affect the validity of Theorem~\ref{th:muf-to-ra}.)
\begin{align*}
  Q' &= Q'_Q \cup Q'_\delta \cup \{\TT\}, \\
     &\
     \begin{aligned}[t]
       & Q'_Q =\{\sigma(V_q) \mid q\in Q\}, \
         Q'_\delta = \{\sigma(V_r) \mid r\in\delta\},
     \end{aligned} \\
  q'_0 &= \sigma(V_{q_0}), \\
  \delta' &
  \begin{aligned}[t]
    {}={}   &\{ (\TT,\TT,\emptyset)\TO\TT \}\\
    {}\cup{}&\{ (\sigma(V_q),\varepsilon,\emptyset)\TO\sigma(V_r)
                \mid q\in Q,\ r\in\delta|_q\} \\
    {}\cup{}&\{ (\sigma(V_r),\phi,R)\TO\sigma(V_{q'})
                \mid r=(q,\phi,R)\TO q'\in\delta \},
  \end{aligned}
  \\
  F' &= \{ \sigma(V_q)\mid q\in F \} \cup \{\TT\}.
\end{align*}

It is sufficient to show $L(\calA)=L(\calA')$.
We prove the following more general statement:%
\footnote{
  Since state $\TT$ of $\calA'$ is not reachable from
  the initial state,
  an accepting computation of $\calA'$ should visit
  state $\sigma(V_q)$ for some $q\in F$
  infinitely often, and thus
  statement~(\ref{eq:ra-to-muf}) is sufficient to show
  $L(\calA)=L(\calA')$.
}
\begin{equation}
\begin{aligned}
& \forall w, q,q', i,j, \theta, \theta'
 \text{ such that } i\le j\colon\\
& 
  \quad
  (q,\theta,\Rangew{i}{})\vdash_{\!\calA}^* (q',\theta',\Rangew{j}{})
 \\ &\quad {}\Leftrightarrow 
  (\sigma(V_q),\theta,\Rangew{i}{})\vdash_{\!\calA'}^*
  (\sigma(V_{q'}),\theta',\Rangew{j}{}).
\end{aligned}
\label{eq:ra-to-muf}
\end{equation}
Note that
for $q,q'\in Q$, $q\ne q'$ implies
$\sigma(V_q)\neq\sigma(V_{q'})$,
since
$\delta|_q$ and $\delta|_{q'}$ are disjoint
if $q\ne q'$.

$(\Rightarrow)$
It is sufficient to show that
$(q,\theta,\Rangew{i}{})\vdash_{\!\calA} (q',\theta',\Rangew{j}{})$
implies
$(\sigma(V_q),\theta,\Rangew{i}{})\vdash_{\!\calA'}^*
(\sigma(V_{q'}),\theta',\Rangew{j}{})$.
This can be easily proved by the contents of $\delta'$.

$(\Leftarrow)$
Since we can split a transition sequence
$(\sigma(V_q),\theta,\Rangew{i}{})\vdash_{\!\calA'}^+
(\sigma(V_{q'}),\theta',\Rangew{j}{})$
at every visit to an element of $Q'_Q$,
it is sufficient to prove statement~(\ref{eq:ra-to-muf})
for a transition sequence
where no element of $Q'_Q$ appears
except the first and last IDs.
To enable this transition sequence,
$\calA'$ must have transition rules
$(\sigma(V_q),\varepsilon,\emptyset)\to\sigma(V_r)$
and
$(\sigma(V_r),\phi,R)\to\sigma(V_{q'})$
for some $r=(q,\phi,R)\to q' \in\delta$
such that $w,i,\theta\models\phi$ and
$(\theta',j) = (\theta[R\Gets\Snd{w_i}],i+1)$.
Hence, $(q,\theta,\Rangew{i}{})\vdash_{\!\calA}^*
(q',\theta',\Rangew{j}{})$.
\QED
\end{proof}

\section{Conclusion}
\label{sec:conclusion}

In this paper,
we proposed
$\mufo$-calculus,
which is a subclass of $\mu$-calculus with the freeze quantifier and
has the same expressive power as
\Buchi\ RA\@.
We showed a transformation from a system of equations in $\mufo$-calculus
to an equivalent \Buchi\ RA as well as its reverse.
We proved the correctness of the transformations.
We have also provided formal proofs of those theorems
using the Coq proof assistant.%
\footnote{%
The proof scripts are available at
\url{https://github.com/ytakata69/proof-mucal-bra}.}

Future work includes
developing a formal verification method
based on RA and the proposed calculus.

\bibliographystyle{ieicetr}
\bibliography{rpds}

\end{document}